\documentclass{llncs}

\usepackage{array}
\usepackage{float}
\usepackage{graphicx}

\makeatletter

\newcommand{\noun}[1]{\textsc{#1}}
\providecommand{\tabularnewline}{\\}
\floatstyle{ruled}
\newfloat{algorithm}{tbp}{loa}
\floatname{algorithm}{Algorithm}

\spnewtheorem{observation}{Observation}{\bfseries}{\itshape}
\spnewtheorem{rrule}{Rule}{\bfseries}{\rmfamily}
\spnewtheorem*{rrule*}{Rule}{\bfseries}{\rmfamily}
\spnewtheorem{fact}{Fact}{\bfseries}{\rmfamily}
\spnewtheorem{numclaim}{Claim}{\itshape}{\rmfamily}
\spnewtheorem*{lemma*}{Lemma}{\bfseries}{\itshape}

\AtBeginDocument{

}

\newcommand{\bem}[1]{{\bf {\em #1}}}

\makeatother

\begin{document}

\title{Solving Dominating Set in Larger Classes of Graphs: FPT Algorithms
and Polynomial Kernels}

\author{Geevarghese Philip \and Venkatesh Raman \and Somnath Sikdar}

\institute{The Institute of Mathematical Sciences, Chennai, India.\\
\email{\{gphilip|vraman|somnath\}@imsc.res.in}}

\maketitle
\begin{abstract}
We show that the \noun{$k$-Dominating Set} problem is fixed parameter
tractable (FPT) and has a polynomial kernel for any class of graphs
that exclude $K_{i,j}$ as a subgraph, for any fixed $i,j\ge1$. This
strictly includes every class of graphs for which this problem has
been previously shown to have FPT algorithms and/or polynomial kernels.
In particular, our result implies that the problem restricted to bounded-degenerate
graphs has a polynomial kernel, solving an open problem posed by Alon
and Gutner in~\cite{AlonGutner2008}.
\end{abstract}

\section{Introduction}

The \noun{$k$-Dominating Set} problem asks, for a graph $G=(V,E)$
and a positive integer $k$ given as inputs, whether there is a vertex-subset
$S\subseteq V$ of size at most $k$ such that every vertex in $V\setminus S$
is adjacent to some vertex in $S$. Such a vertex-subset is called
a \emph{dominating set} of $G$. This problem is known to be NP-hard
even in very restricted graph classes, such as the class of planar
graphs with maximum degree~$3$~\cite{GareyJohnson1979}. In the
world of parameterized complexity, this is one of the most important
hard problems: the problem parameterized by $k$ is the canonical
$W\left[2\right]$-hard problem~\cite{DowneyFellows1999}. The problem
remains $W\left[2\right]$-hard even in many restricted classes of
graphs --- for example, it is $W\left[2\right]$-hard in classes of
graphs with bounded average degree~\cite{GolovachVillanger2008}.
This latter fact implies that it is unlikely that the problem has
a fixed-parameter-tractable (FPT) algorithm on graphs with a bounded
average degree, that is, an algorithm that runs in time $f(k)\cdot n^{c}$
for \emph{some} computable function $f(k)$ independent of the input
size $n$, and some constant $c$ independent of $k$. %
\footnote{To know more about the notions of FPT and $W$-hardness and to see
why it is considered unlikely that a $W\left[2\right]$-hard problem
will have an FPT algorithm, see~\cite{DowneyFellows1999}.%
}

The problem has an FPT algorithm on certain restricted families of
graphs, such as planar graphs~\cite{FominThilikos2006}, graphs of
bounded genus~\cite{EllisFanFellows2004}, $K_{h}$-topological-minor-free
graphs, and graphs of bounded degeneracy~\cite{AlonGutner2007};
these last being, to the best of our knowledge, the most general graph
class previously known to have an FPT algorithm for this problem.
In the current paper, we show that the problem has an FPT algorithm
in a class of graphs that encompasses, and is strictly larger than,
all such previously known classes --- namely, the class of $K_{i,j}$-free
graphs.

Closely related to the notion of an FPT algorithm is the concept of
a \emph{kernel} for a parameterized problem. For the \noun{$k$-Dominating
Set} problem parameterized by $k$, a kernelization algorithm is a
polynomial-time algorithm that takes $(G,k)$ as input and outputs
a graph $G'$ and a nonnegative integer $k'$ such that the size of
$G'$ is bounded by some function $g(k)$ of $k$ alone, $k'\le h(k)$
for some function $h(k)$, and $G$ has a dominating set of size at
most $k$ if and only if $G'$ has a dominating set of size at most
$k'$. The resulting instance $G'$ is called a kernel for the problem.
A parameterized problem has a kernelization algorithm if and only
if it has an FPT algorithm~\cite{DowneyFellows1999}, and so it is
unlikely that the \noun{$k$-Dominating Set} problem on general graphs
or on graphs having a bounded average degree has a kernelization algorithm.
For the same reason, this problem has a kernelization algorithm when
restricted to those graph classes for which it has an FPT algorithm.
But the size of the kernel obtained from such an algorithm could be
exponential in $k$, and finding if the kernel size can be made smaller
--- in particular, whether it can be made polynomial in $k$ --- is
an important problem in parameterized complexity.

Proving polynomial bounds on the size of the kernel for different
parameterized problems has been a significant practical aspect in
the study of the parameterized complexity of NP-hard problems, and
many positive results are known. See~\cite{GuoNiedermeier2007} for
a survey of kernelization results. 

For the \noun{$k$-Dominating Set} problem, the first polynomial kernel
result was obtained by Alber et al.~\cite{AlberFellowsNiedermeier2004}
in 2004: they showed that the problem has a \emph{linear} kernel of
at most $335k$ vertices in planar graphs. This bound for planar graphs
was later improved by Chen et al.~\cite{ChenFernauKanjXia2007} to
$67k$. Fomin and Thilikos~\cite{FominThilikos2004} showed in 2004
that the same reduction rules as used by Alber et al.\ give a linear
kernel (linear in $k+g$) for the problem in graphs of genus $g$.
The next advances in kernelizing this problem were made by Alon and
Gutner in 2008~\cite{AlonGutner2008}. They showed that the problem
has a linear kernel in $K_{3,h}$-topological-minor-free graph classes
(which include, for example, planar graphs), and a polynomial kernel
(where the exponent depends on $h$) for $K_{h}$-topological-minor-free
graph classes, these last being the most general class of graphs for
which the problem has been previously shown to have a polynomial kernel.
In the meantime, the same authors had shown in 2007 that the problem
is FPT in (the strictly larger class of) graphs of bounded degeneracy~\cite{AlonGutner2007},
but had left open the question whether the problem has a polynomial
kernel in such graph classes. In this paper, we answer this question
in the affirmative, and show, in fact, that a strictly larger classes
of graphs --- the $K_{i,j}$-free graph classes --- have polynomial
kernels for this problem. 

See Table~\ref{tab:KnownResults} for a summary of some FPT and kernelization
results for the \noun{$k$-Dominating Set} problem on various classes
of graphs.%
\begin{table}
\begin{centering}
\begin{tabular}{|c|@{\hspace{2mm}}@{\hspace{2mm}}l|@{\hspace{2mm}}@{\hspace{2mm}}l@{\hspace{2mm}}|}\hline \hline 
\emph{Graph Class}  & \emph{FPT Algorithm Running Time} & \emph{Kernel Size}\tabularnewline \hline 
 &  & \tabularnewline
Planar  & $O(k^{4}+2^{15.13\sqrt{k}}k+n^{3})$~\cite{FominThilikos2006} & $O(k)$~\cite{AlberFellowsNiedermeier2004,ChenFernauKanjXia2007} \tabularnewline
 &  & \tabularnewline
Genus-$g$  & $O((24g^{2}+24g+1)^{k}n^{2})$~\cite{EllisFanFellows2004} & $O(k+g)$~\cite{FominThilikos2004} \tabularnewline
 &  & \tabularnewline
$K_{h}$-minor-free & $2^{O(\sqrt{k})}n^{c}$~\cite{DemaineFominHajiaghayiThilikos2005:2},$O((\log h))^{hk/2}\cdot n$~\cite{AlonGutner2007} & $O(k^{c})$~\cite{AlonGutner2008} \tabularnewline
 &  & \tabularnewline
$K_{h}$-topological-minor-free  & $(O(h))^{hk}\cdot n$~\cite{AlonGutner2007} & $O(k^{c})$~\cite{AlonGutner2008}\tabularnewline
 &  & \tabularnewline
$d$-degenerate  & $k^{O(dk)}n$~\cite{AlonGutner2007}  & $k^{O(dk)}$~\cite{AlonGutner2007}, $O((d+2)^{2(d+2)}\cdot k^{2(d+1)^{2}})^{\dagger}$\tabularnewline
 &  & \tabularnewline
$K_{i,j}$-free & $O(n^{i+O(1)})^{\dagger}$ & $O((j+1)^{2(i+1)}k^{2i^{2}})^{\dagger}$\tabularnewline
 &  & \tabularnewline
\hline \hline
\end{tabular}
\par\end{centering}

\caption{\label{tab:KnownResults}Some FPT and kernelization results for $k$\noun{-Dominating
Set.} Results proved in this paper are marked with a \noun{$\dagger$.}}

\end{table}

\paragraph{Our Results.}

We show that for any fixed $i,j\ge1$, the \noun{$k$-Dominating Set}
problem has a polynomial kernel on graphs that do not have~$K_{i,j}$
(a complete bipartite graph with the two parts having~$i$ and~$j$
vertices) as a subgraph. For input graph~$G$ and parameter~$k$,
the size of the kernel is bounded by~$k^{c}$ where~$c$ is a constant
that depends only on~$i$ and~$j$. A graph~$G$ is said to be
$d$\emph{-degenerate} if every subgraph of~$G$ has a vertex of
degree at most~$d$. Since a $d$-degenerate graph does not have~$K_{d+1,d+1}$
as a subgraph, it follows that the \noun{$k$-Dominating Set} problem
has a polynomial kernel on graphs of bounded degeneracy. This settles
a question posed by Alon and Gutner in~\cite{AlonGutner2008}. We
also provide a slightly simpler and a smaller kernel for the version
where we want the \noun{$k$-Dominating Set} to be independent as
well.

Note that except for $d$-degenerate graphs, all the other graph classes
in \mbox{Table~\ref{tab:KnownResults}} are minor-closed. This
seems to be indicative of the state of the art --- the only other
previous FPT or kernelization result for the \noun{$k$-Dominating
Set} problem on a non-minor-closed class of graphs that we know of
is the $O(k^{3})$ kernel and the resulting FPT algorithm for graphs
that exclude triangles and $4$-cycles~\cite{RamanSaurabh2008}.
In fact, this result can be modified to obtain similar bounds on graphs
with just no $4$-cycles (allowing triangles). Since a $4$-cycle
is just $K_{2,2}$, this result follows from the main result of this
paper by setting $i=j=2$.

Since a $K_{h}$-topological-minor-free graph has bounded degeneracy~\cite[Proposition 3.1]{AlonGutner2008}
(for a constant $h$), the class of $K_{i,j}$-free graphs is more
general than the class of $K_{h}$-topological-minor-free graphs.
Thus we extend the class of graphs for which the \noun{$k$-Dominating
Set} problem is known to have (1)~FPT algorithms and (2)~polynomial
kernels, to the class of $K_{i,j}$-free graphs.

\paragraph{Organization of the rest of the paper.}

In Section~\ref{sec:kij-free kernel}, we develop our main algorithm
that runs in $O(n^{i+O(1)})$ time and constructs a kernel of size
$O((j+1)^{2\left(i+1\right)}k^{2i^{2}})$ for $k$\noun{-Dominating
Set} on $K_{i,j}$-free graphs, for fixed $j\ge i\ge2$. As a corollary
we obtain, in Section~\ref{sec:kernel_d_degenerate}, a polynomial
kernel for $d$-degenerate graphs, with running time $O(n^{O(d)})$
and kernel size $O((d+2)^{2(d+2)}k^{2(d+1)^{2}})$. In Section~\ref{sub:faster_d_degenerate}
we describe an improvement to the above algorithm that applies to
$d$-degenerate input graphs which yields a kernel of the same size
as above and runs in time $O(2^{d}dn^{2})$. In Section~\ref{sec:indep_dom_set_ij}
we describe a modification of the algorithm in Section~\ref{sec:kij-free kernel}
that constructs a polynomial kernel for the $k$-\noun{Independent
Dominating Set} problem on $K_{i,j}$-free graphs. This kernel has
$O(jk^{i})$ vertices, and so implies a kernel of size $O((d+1)k^{d+1})$
for this problem on $d$-degenerate graphs. In Section~\ref{sec:Conclusion}
we state our conclusions and list some open problems.

\paragraph{Notation.}

All the graphs in this paper are finite, undirected and simple. In
general we follow the graph terminology from~\cite{Diestel2000}.
We let $V(G)$ and $E(G)$ denote, respectively, the vertex and edge
sets of a graph $G$. The \emph{open-neighborhood} of a vertex $v$
in a graph $G$, denoted $N(v)$, is the set of all vertices that
are adjacent to $v$ in $G$. A \emph{$k$-dominating set} of graph
$G$ is a vertex-subset $S$ of size at most $k$ such that for each
$u\in V(G)\setminus S$ there exists $v\in S$ such that $\{u,v\}\in E(G)$.
Given a graph $G$ and $A,B\subseteq V(G)$, we say that $A$ dominates
$B$ if every vertex in $B\setminus A$ is adjacent in $G$ to some
vertex in $A$.

\section{\label{sec:kij-free kernel} A Polynomial Kernel for $K_{i,j}$-free
Graphs}

In this section we consider the parameterized $k$\noun{-Dominating
Set} problem on graphs that do not have $K_{i,j}$ as a subgraph,
for fixed $j\ge i\ge1$. It is easy to see that the problem has a
linear kernel when $i=1,j\ge i$, so we consider the cases $j\ge i\ge2$.
We solve a more general problem, namely the \noun{rwb-Dominating Set}
problem, defined as follows: Given a graph $G$ whose vertex set $V$
is partitioned into $R_{G},W_{G}$, and $B_{G}$ (colored red, white,
and black, respectively) and a non-negative integer parameter $k$,
is there a subset $S\subseteq V$ of size at most $k$ such that $R_{G}\subseteq S$
and $S$ dominates $B_{G}$? We call such an $S$ an \emph{rwb-dominating}
set of $G$, and such a graph an \emph{rwb-graph}.

Intuitively, the vertices colored red are those that will be picked
up by the reduction rules in the $k$-dominating set $D$ that we
are trying to construct. In particular, if there is a $k$-dominating
set in the graph, there will be one that contains all the red vertices.
White vertices are those that have been already dominated. Clearly
all neighbors of red vertices are white, but our reduction rules color
some vertices white even if they have no red neighbors (at that point).
These are vertices that will be dominated by one of some constant
number of vertices identified by the reduction rules. See reduction
rule 2 for more details. Black vertices are those that are yet to
be dominated. It is easy to see that if we start with a general graph
$G$ and color all its vertices black to obtain an rwb-graph $G'$,
then $G$ has a dominating set of size at most $k$ if and only if
$G'$ has an rwb-dominating set of size at most $k$.

We first describe an algorithm that takes as input an rwb-graph $G$
on $n$ vertices and a positive number $k$, and runs in $O(n^{i+O(1)})$
time. The algorithm either finds that $G$ does not have any rwb-dominating
set of size at most $k$, or it constructs an instance $(G',k')$
on $O((j+1)^{i+1}k^{i^{2}})$ vertices such that~$G$ has an rwb-dominating
set of size at most $k$ if and only if $G'$ has an rwb-dominating
set of size at most $k'$.

The algorithm applies a sequence of reduction rules in a specified
order. The input and output of each reduction rule are rwb-graphs. 

\begin{definition}
\label{def:reduced_graph}We say that graph $G$ is \bem{reduced}
with respect to a reduction rule if an application of the rule to
$G$ does not change $G$. 
\end{definition}
Each reduction rule satisfies the following correctness condition
and preserves the invariants stated below:

\begin{definition}
\label{def:rule_correctness_ij}\emph{(Correctness)} A reduction rule
$R$ is said to be \bem{correct} if the following condition holds:
if $(G',k')$ is the instance obtained from $(G,k)$ by one application
of rule $R$ then $G'$ has an rwb-dominating set $D'$ of size $k'$
if and only if $G$ has an rwb-dominating set $D$ of size $k$. %
\footnote{\label{fn:k-no-change}Note, however, that none of our reduction rules
changes the value of $k$, and so $k'=k$ for every one of these rules.%
}
\end{definition}
\label{rem:reduction_rule_properties_ij}\textbf{Invariants}:

\begin{enumerate}
\item \label{rem:no_new_kij}None of the reduction rules introduces a $K_{i,j}$
into a graph.
\item \label{rem:red_nbrs_white_ij}In the rwb-graphs constructed by the
algorithm, red vertices have all white neighbors.
\item \label{rem:rules_no_regress_ij}Let $R$ be any reduction rule, and
let $R'$ be a rule that precedes $R$ in the given order. If $G$
is a graph that is reduced with respect to $R'$ and $G'$ is a graph
obtained by applying $R$ to $G$, then $G'$ is reduced with respect
to $R'$.
\end{enumerate}

\subsection{\label{sub:rules_ij}The reduction rules and the kernelization algorithm}

The kernelization algorithm assumes that the input graph is an rwb-graph.
It applies the following rules exhaustively in the {\em given order}.
Each rule is repeatedly applied till it causes no changes to the graph
and then the next rule is applied.

To make it easier to present the reduction rules and the arguments
of their correctness, we use a couple of notational conventions in
this section. For each rule below, $G$ denotes the graph on which
the rule is applied, and $G'$ the graph that results. Further, $D$
and $D'$ are as in Definition~\ref{def:rule_correctness_ij}: $D$
is an rwb-dominating set of size $k$ of $G$, and $D'$ an rwb-dominating
of $G'$ of size $k'$. $^{\ref{fn:k-no-change}}$

\begin{rrule}\label{rul:one_ij}Color all isolated black vertices
of $G$ red. \end{rrule}

Rule~\ref{rul:one_ij} is correct as the only way to dominate isolated
black vertices is by picking them in the proposed rwb-dominating set.

\begin{rrule}\label{rul:two_ij} For $p=1,2,\ldots,i-2$, in this
order, apply Rule $2.p$ repeatedly till it no longer causes any changes
in the graph.

\textbf{Rule}~$\mathbf{2.p}$ Let $b=jk$ if $p=1$, and $b=jk^{p}+k^{p-1}+k^{p-2}\cdots+k$
if $2\le p\le i-2$. If a set of $(i-p)$ vertices $U=\{u_{1},u_{2},\ldots,u_{i-p}\}$,
none of which is red, has more than $b$ common black neighbors, then
let $B$ be this set of neighbors. 

\begin{enumerate}
\item Color all the vertices in $B$ white. 
\item Add to the graph $(i-p)$ new (gadget) vertices $X=\{x_{1},x_{2},\ldots,x_{i-p}\}$
and all the edges $\{u,x\};u\in U,x\in X$, as in Figure~\ref{fig:Rule_2_ij}. 
\item Color all the vertices in $X$ black. 
\end{enumerate}
\end{rrule}

\begin{figure}
\includegraphics[bb=-350bp -5bp 770bp 466bp,scale=0.22]{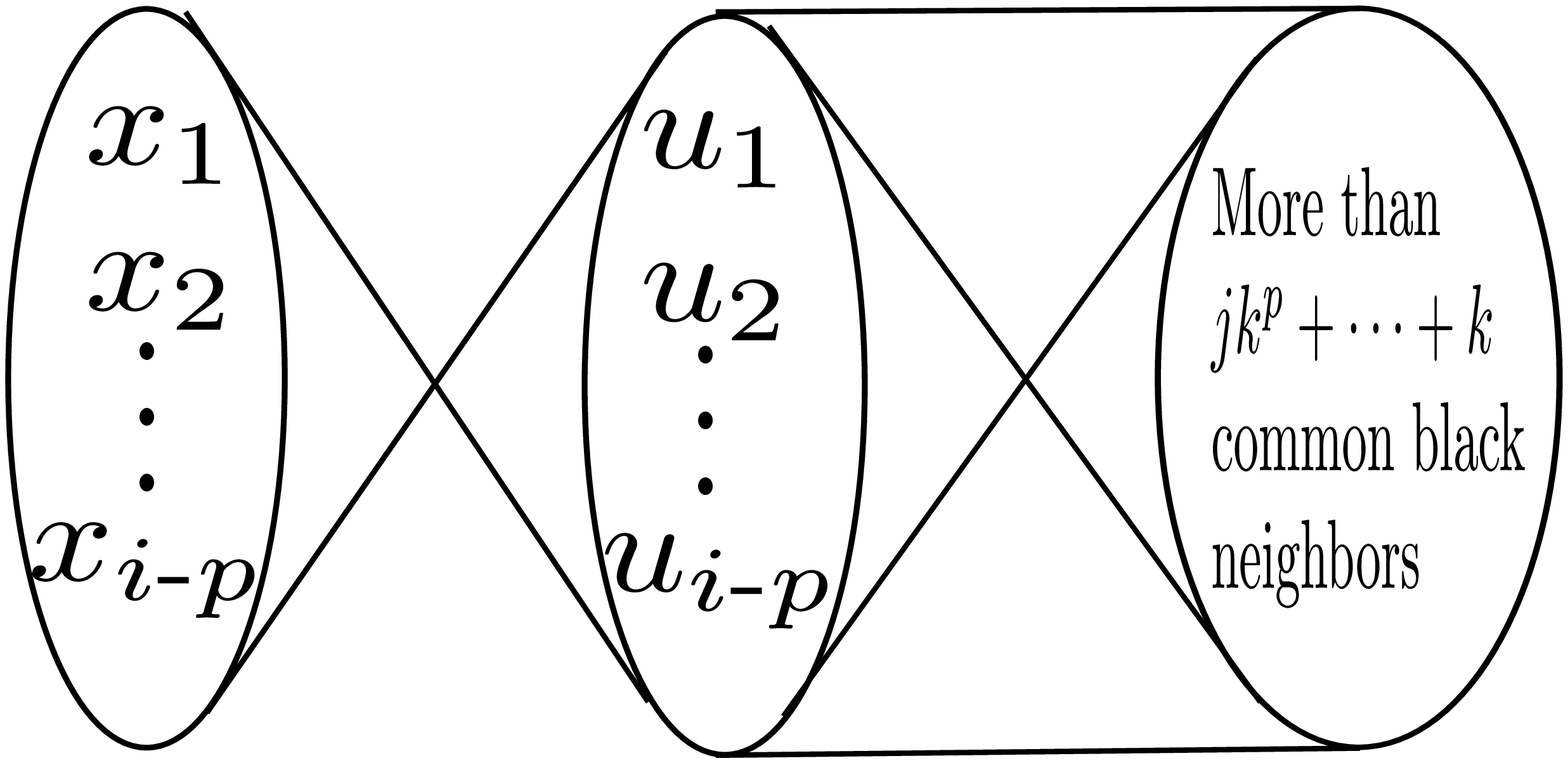} 

\caption{\label{fig:Rule_2_ij}Rule~\ref{rul:two_ij}}

\end{figure}

\begin{numclaim}\label{cla:rule2p_forces_one_ij}Consider the application
of Rule~$2.p,1\le p\le i-2$. If $U$ is a set of vertices of $G$
that satisfies the condition in Rule $2.p$, then at least one vertex
in $U$ must be in any subset of $V(G)$ of size at most $k$ that
dominates $B$. \end{numclaim} 

\begin{proof}
We give a proof when $p=1$. The proof for larger values of $p$ is
along similar lines by reducing it to the case for smaller values
of $p$ as in the proof of Claim~\ref{cla:forced_red_single_ij}
below.

When $p=1$, suppose that there is a rwb-dominating set $D$ of $G$
of size at most $k$ that does not contain any vertex of $U$. Since
$U$ has more than $b=jk$ common black neighbors, there is a vertex
in $D$ that dominates at least $j+1$ common black neighbors of $U$
(possibly including itself). That vertex along with $U$ forms a $K_{i,j}$
in $G$, contradicting either the property of the input graph or the
first invariant for the rules. \qed 
\end{proof}
\begin{lemma}
\label{lem:rule2_correct_ij} Rule~$2.p$ is correct for $1\le p\le i-2$. 
\end{lemma}
\begin{proof}
If $G$ has an rwb-dominating set $D$ of size $k$, then $D\cap U\ne\emptyset$
by Claim~\ref{cla:rule2p_forces_one_ij}. So $D':=D$ is an rwb-dominating
set of $G'$, since $D\cap U$ dominates $X$. For the other direction,
assume that $D'$ exists. If $D'\cap U=\emptyset$ then since $D'$
dominates $X$ and $X$ is independent, $X\subseteq D'$, and so set
$D:=D'\setminus X\cup U$. If $D'\cap X=\emptyset$ then since $D'$
dominates $X$, $D'\cap U\ne\emptyset$, and so set $D:=D'$. If $D'\cap U\ne\emptyset$
and $D'\cap X\ne\emptyset$ then pick an arbitrary vertex $b\in B$
and set $D:=D'\setminus X\cup\{b\}$.\qed 
\end{proof}
\begin{rrule}\label{rul:three_ij}If a black or white vertex $u$
has more than $jk^{i-1}+k^{i-2}+\cdots+k^{2}+k$ black neighbors,
then color $u$ red and color all the black neighbors of $u$ white.
\end{rrule}

\begin{numclaim}\label{cla:forced_red_single_ij} Let $G$ be reduced
with respect to Rule~\ref{rul:one_ij} and Rules~$2.1$ to~$2.(i-2)$.
If a black or white vertex $u$ of $G$ has more than $h=jk^{i-1}+k^{i-2}+\cdots+k^{2}+k$
black neighbors (let this set of neighbors be $B$), then $u$ must
be in any subset of $V(G)$ of size at most $k$ that dominates $B$.
\end{numclaim} 

\begin{proof}
Let $S\subseteq V(G)$ be a set of size at most $k$ that dominates
$B$. If $S$ does not contain $u$, then there is a $v\in S$ that
dominates at least $(h/k)+1$ of the vertices in $B$. The vertex
$v$ is not red (because of the second invariant), and $u,v$ have
$h/k>jk^{i-2}+k^{i-3}+\cdots+1$ common black neighbors, a contradiction
to the fact that $G$ is reduced with respect to Rule~$2.(i-2)$.
\qed 
\end{proof}
This proves the correctness of Rule~\ref{rul:three_ij} on graphs
reduced with respect to Rule~\ref{rul:one_ij} and Rules~$2.1$
to~$2.(i-2)$.

\begin{rrule}\label{rul:four_ij}If a white vertex \emph{$u$} is
adjacent to at most one black vertex, then delete $u$ and apply Rule~\ref{rul:one_ij}.
\end{rrule}

It is easy to see that Rule~\ref{rul:four_ij} is correct, since
if $u$ has no black neighbor in $G$ then $u$ has no role in dominating
$B_{G}$; if $u$ has a single black neighbor $v$ then we can replace
$u$ with $v$ in $D'$.

\begin{rrule}\label{rul:five_ij}If there is a white vertex $u$
and a white or black vertex $v$ such that $N(u)\cap B_{G}\subseteq N(v)\cap B_{G}$,
then delete $u$ and apply Rule~\ref{rul:one_ij}.

\end{rrule}

The correctness of this rule follows from the fact that if $D$ chooses
$u$, then we can choose $v$ in $D'$.

\begin{rrule}\label{rul:six_ij}If $|R_{G}|>k$ or $|B_{G}|>jk^{i}+k^{i-1}+k^{i-2}+\cdots+k^{2}$
then output {}``NO''. \end{rrule}

The correctness of the rule when $|R_{G}|>k$ is obvious as the proposed
dominating set we construct should contain all of $R_{G}$. Note that
in a graph $G$ reduced with respect to Rules~1 to 5, no white or
black vertex has more than $jk^{i-1}+k^{i-2}+\cdots+k$ black neighbors,
or else Rule~\ref{rul:three_ij} would have applied, contradicting
the third invariant. Hence $k$ of these vertices can dominate at
most $jk^{i}+k^{i-1}+k^{i-2}+\cdots+k^{2}$ black vertices and hence
if $|B_{G}|>jk^{i}+k^{i-1}+k^{i-2}+\cdots+k^{2}$, the algorithm is
correct in saying {}``NO''.

\subsection{Algorithm correctness and kernel size\label{sub:kernel_algorithm_correctness_ij}}

The following claim giving the correctness of the kernelization algorithm
follows from the correctness of the reduction rules.

\begin{numclaim}\label{cla:alg_preserve_k_ij}Let $G$ be the input
rwb-graph and $H$ the rwb-graph constructed by the algorithm after
applying all the reduction rules. Then $G$ has an rwb-dominating
set of size at most $k$ if and only if there is an rwb-dominating
set of size at most $k$ in $H$. \end{numclaim}

Now we move on to prove a polynomial bound on the size of the reduced
instance. 

\begin{lemma}
\label{lem:algorithm_correctness_ij} Starting with a $K_{i,j}$-free
rwb-graph $G$ as input, if the kernelization algorithm says {}``NO''
then $G$ does not have an rwb-dominating set of size at most $k$.
Otherwise, if the algorithm outputs the rwb-graph $H$, then $|V(H)|=O((j+1)^{i+1}k^{i^{2}})$. 
\end{lemma}
\begin{proof}
The correctness of the Rule~\ref{rul:six_ij} establishes the claim
if the algorithm says {}``NO''. Now suppose the algorithm outputs
$H$ without saying {}``NO''. The same rule establishes that $|R_{H}|\le k$
and $b=|B_{H}|\le jk^{i}+k^{i-1}+\cdots+k\le(j+1)k^{i}$. Now we bound
$|W_{H}|$. Note that no two white vertices have identical black neighborhoods,
or else Rule~\ref{rul:five_ij} would have applied. Also each white
vertex has at least two black neighbors, or else Rule~\ref{rul:four_ij}
would have applied. Hence the number of white vertices that have less
than $i$ black neighbors is at most ${b \choose 2}+{b \choose 3}+\cdots+{b \choose i-1}\le2b^{i-1}$.
No set of $i$ black vertices has more than $(j-1)$ common white
neighbors, or else these form a $K_{i,j}$. Hence the number of white
vertices that have $i$ or more black neighbors is at most ${b \choose i}(j-1)\le(j-1)b^{i}$.
The bound in the lemma follows.\qed 
\end{proof}
The algorithm can be implemented in $O(n^{i+O(1)})$ time, as the
main Rule~$2$ can be applied by running through various subsets
of $V(G)$ of size $p$ for $p$ ranging from~$1$ to $i-2$. Thus,
we have 

\begin{lemma}
\label{lem:kij_poly_kernel_rwb} For any fixed $j\ge i\ge1$, the
\noun{rwb-Dominating Set} problem (with parameter $k$) on $K_{i,j}$-free
graphs has a polynomial kernel with $O((j+1)^{i+1}k^{i^{2}})$ vertices. 
\end{lemma}
To obtain a polynomial kernel for the \noun{$k$-Dominating Set} problem
on $K_{i,j}$-free graphs, we first color all the vertices black and
use Lemma~\ref{lem:kij_poly_kernel_rwb} on this \noun{rwb-Dominating
Set} problem instance. To transform the reduced colored instance $H$
to an instance of (the uncolored) \noun{$k$-dominating Set}, we can
start by deleting all the red vertices, since they have no black neighbors.
But we need to capture the fact that the white vertices need not be
dominated. This can be done by, for example, adding a new vertex $v_{x}$
adjacent to every vertex $x$ in $W_{H}$ of the reduced graph $H$,
and attaching $k+|W_{H}|+1$ separate pendant vertices to each of
the vertices $v_{x}$. It is easy to see that the new graph does not
have a $K_{i,j}$, $j\geq i\geq2$, if $H$ does not have one and
that $H$ has at most $k$ black or white vertices dominating $B_{H}$
if and only if the resulting (uncolored) graph has a dominating set
of size at most $|W_{H}|+k$. Thus after reducing to the uncolored
version, $k$ becomes $k+|W_{H}|$ and the number of vertices increases
by $(k+|W_{H}|+2)\cdot|W_{H}|$. Hence by Lemma~\ref{lem:kij_poly_kernel_rwb},
we have 

\begin{theorem}
\label{thm:kij_poly_kernel} For any fixed $j\ge i\ge1$, the \noun{$k$-Dominating
Set} problem on $K_{i,j}$-free graphs has a polynomial kernel with
$O((j+1)^{2(i+1)}k^{2i^{2}})$ vertices. 
\end{theorem}

\section{\label{sec:kernel_d_degenerate}A Polynomial Kernel for $d$-degenerate
Graphs}

A $d$-degenerate graph does not contain $K_{d+1,d+1}$ as a subgraph,
and so the kernelization algorithm of the previous section can be
applied to a $d$-degenerate graph, setting $i=j=d+1$. The algorithm
runs in time $O((d+1)^{2}n^{d+O(1)})$ and constructs a kernel with~$O((d+2)^{2(d+2)}\cdot k^{2(d+1)^{2}})$
vertices. Since a $d$-degenerate graph on $v$ vertices has at most
$dv$ edges, we have: 

\begin{corollary}
\label{cor:d_degenerate_poly_kernel}The \noun{$k$-Dominating Set}
problem on $d$-degenerate graphs has a kernel on~$O((d+2)^{2(d+2)}\cdot k^{2(d+1)^{2}})$
vertices and edges. 
\end{corollary}
Corollary~\ref{cor:d_degenerate_poly_kernel} settles an open problem
posed by Alon and Gutner in~\cite{AlonGutner2008}.

\subsection{Improving the running time\label{sub:faster_d_degenerate}}

We describe a modification of our algorithm to $d$-degenerate graphs
that makes use of the following well known property of $d$-degenerate
graphs, to reduce the running time to $O(2^{d}\cdot dn^{2})$; the
bound on the kernel size remains the same.

\begin{fact}\label{fac:degenerate_ordering}\cite[Theorem 2.10]{FranceschiniLuccioPagli2006}
Let $G$ be a $d$-degenerate graph on $n$ vertices. Then one can
compute, in $O(dn)$ time, an ordering $v_{1},v_{2},\ldots,v_{n}$
of the vertices of $G$ such that for $1\le i\le n$, $v_{i}$ has
at most $d$ neighbors in the subgraph of $G$ induced on $\{v_{i+1},\ldots,v_{n}\}$.\end{fact}

The modification to the algorithm pertains to the way rules~$2.1$
to~$2.(d-1)$ are implemented: the rest of the algorithm remains
the same.

In implementing Rule~$2.p,1\le p\le(d-1)$, instead of checking each
$(d-p+1)$-subset of vertices in the graph to see if it satisfies
the condition in the rule, we make use of Fact~\ref{fac:degenerate_ordering}
to quickly find such a set of vertices, if it exists. Let $G$ be
the graph instance on $n$ vertices on which Rule $2.p$ is to be
applied. First we delete, temporarily, all the red vertices in $G$.
We then find an ordering $v_{1},v_{2},\ldots,v_{n}$ of the kind described
in the above fact, of all the remaining vertices in $G$. Let $U$
and $B$ be as defined in the rule. The first vertex $v_{l}$ in $U\cup B$
that appears in the ordering has to be from $B$, since each vertex
in $U$ has degree greater than $d$. The vertex $v_{l}$ will then
have a neighborhood of size $d-p+1$ that in turn has $B$ as its
common neighborhood. We use this fact to look for such a pair $(U,B)$
and exhaustively apply Rule~$2.p$ to $G$. See Algorithm~\ref{alg:faster_d_degenerate}
for pseudocode of the algorithm. We then add back the red vertices
that we deleted prior to this step, along with all their edges to
the rest of the graph.

\begin{algorithm}
\texttt{\small for $l:=1$ to $n$}{\small \par}

\texttt{\small do}{\small \par}

\texttt{\small ~~if $v_{l}$ is black and its degree in $G[v_{l+1},\ldots,v_{n}]$
is at least $d-p+1$}{\small \par}

\texttt{\small ~~then}{\small \par}

\texttt{\small ~~~~Find the neighborhood $N$ of $v_{l}$ in $G[v_{l+1},\ldots,v_{n}]$}{\small \par}

\texttt{\small ~~~~for each $(d-p+1)$-subset $S$ of $N$}{\small \par}

\texttt{\small ~~~~do}{\small \par}

\texttt{\small ~~~~~~if $S$ has more than $(d+1)k^{p}+k^{p-1}+\cdots+k$}{\small \par}

\texttt{\small ~~~~~~common black neighbors in $G$}{\small \par}

\texttt{\small ~~~~~~then}{\small \par}

\texttt{\small ~~~~~~~~Apply the three steps of Rule $2.p$,
taking $S$ as $U$}{\small \par}

\texttt{\small ~~~~~~endif}{\small \par}

\texttt{\small ~~~~done}{\small \par}

\texttt{\small ~~endif}{\small \par}

\texttt{\small done}{\small \par}

\caption{\label{alg:faster_d_degenerate}Faster implementation of Rule~$2.p$
in $d$-degenerate graphs.}

\end{algorithm}

As $|N|\le d$, the inner \emph{for} loop is executed at most ${d \choose p-1}$
times for each iteration of the outer loop. Each of the individual
steps in the algorithm can be done in $O(dn)$ time, and so Rule~$2.p$
can be applied in $O(dn\sum_{l=1}^{n}{d \choose p-1})$ time. All
the rules $2.p$ can therefore be applied in~$O(dn\sum_{l=1}^{n}\sum_{p=1}^{d-1}{d \choose p-1})=O(2^{d}\cdot dn^{2})$
time. Thus we have: 

\begin{theorem}
\label{thm:kij_degenerate_poly_kernel}For any fixed $d\ge1$, the
\noun{$k$-Dominating Set} problem on $d$-degenerate graphs has a
kernel on $O((d+2)^{2(d+2)}\cdot k^{(2(d+1))^{2}})$ vertices and
edges, and this kernel can be found in $O(2^{d}\cdot dn^{2})$ time
for an input graph on $n$ vertices. 
\end{theorem}

\section{A polynomial kernel for Independent Dominating Set on $K_{i,j}$-free
graphs}

\label{sec:indep_dom_set_ij}

The $k$\noun{-Independent Dominating Set} problem asks, for a graph
$G$ and a positive integer $k$ given as inputs, whether $G$ has
a dominating set $S$ of size at most $k$ such that $S$ is an independent
set (i.e.\ no two vertices in $S$ are adjacent). This problem is
known to be NP-hard for general graphs~\cite{GareyJohnson1979},
and the problem parameterized by $k$ is $W[2]$-complete~\cite{DowneyFellows1999}.
Using a modified version of the set of reduction rules in Section~\ref{sec:kij-free kernel}
we show that the $k$-\noun{Independent Dominating Set} has a polynomial
kernel in $K_{ij}$-free graphs for $j\ge i\ge1$. For $i=1,j\ge1$
we can easily obtain trivial kernels as before, and for $i=2,j\ge2$
a simplified version of the following algorithm gives a kernel of
size $O(j^{3}k^{4})$.

\subsection{The reduction rules}

Rule~$1$ is the same as for the \noun{Dominating Set} kernel for
$K_{ij}$-free graphs (Section~\ref{sub:rules_ij}). Rules~$2.1$
to $2.(i-2)$ and Rule~$3$ are modified to make use of the fact
that we are looking for a dominating set that is independent. A vertex
$u$ that is made white will never be part of the independent dominating
set $D$ that is sought to be constructed by the algorithm, since
$u$ is adjacent to some vertex $v\in D$. So a vertex can be deleted
as soon as it is made white. Also, rules~$1,2.1\ldots2.(i-2)$ and~$3$
are the only rules. Rules~$4$ and~$5$ from that section do not
apply, because of the same reason as above. The modified rules ensure
that no vertex is colored white, and so they work on \emph{rb-graphs}:
graphs whose vertex set is partitioned into red and black vertices.
Using these modified rules, the bounds of $|R_{H}|$ and $|B_{H}|$
in the proof of Lemma~\ref{lem:algorithm_correctness_ij}, and the
fact that there are no white vertices, we have

\begin{theorem}
\label{thm:kij_ind_dom_poly_kernel}For any fixed $j\ge i\ge1$, the
\noun{$k$-Independent Dominating Set} problem on $K_{i,j}$-free
graphs has a polynomial kernel with $O(jk^{i})$ vertices. 
\end{theorem}
For $d$-degenerate graphs, we have $i=j=d+1$, and therefore we have: 

\begin{corollary}
\label{cor:degenerate_ind_dom_poly_kernel}For any fixed $d\ge1$,
the \noun{$k$-Independent Dominating Set} problem on $d$-degenerate
graphs has a polynomial kernel with $O((d+1)k^{(d+1)})$ vertices. 
\end{corollary}

\section{\label{sec:Conclusion}Conclusions and Future Work}

In this paper, we presented a polynomial kernel for the $k$\noun{-Dominating
Set} problem on graphs that do not have $K_{i,j}$ as a subgraph,
for any fixed $j\ge i\ge1$. We used this to show that the $k$\noun{-Dominating
Set} problem has a polynomial kernel of size $O((d+2)^{2(d+2)}\cdot k^{2(d+1)^{2}})$
on graphs of bounded degeneracy, thereby settling an open problem
from~\cite{AlonGutner2008}. Our algorithm also yielded a slightly
simpler and a smaller kernel for the $k$\noun{-Independent Dominating
Set} problem on $K_{i,j}$-free and $d$-degenerate graphs. These
algorithms are based on simple reduction rules that look at the common
neighborhoods of sets of vertices. 

Dom et al.~\cite{DomLokshtanovSaurabh2008} have shown, by extending
the kernel lower-bound techniques of Bodlaender et al.~\cite{BodlaenderDowneyFellowsHermelin2008},
that the $k$\noun{-Dominating Set} problem on $d$-degenerate graphs
does not have a kernel of size polynomial in \emph{both} $d$ and
$k$ unless the Polynomial Hierarchy collapses to the third level.
This shows that the kernel size that we have obtained for this class
of graphs cannot possibly be significantly improved.

Many interesting classes of graphs are of bounded degeneracy. These
include all nontrivial minor-closed families of graphs such as planar
graphs, graphs of bounded genus, graphs of bounded treewidth, and
graphs excluding a fixed minor, and some non-minor-closed families
such as graphs of bounded degree. Graphs of degeneracy $d$ are $K_{d+1,d+1}$-free.
Since any $K_{i,j};j\ge i\ge2$ contains a~$4$-cycle, every graph
of girth~$5$ is $K_{i,j}$-free. From~\cite[Theorem 1]{Sachs1963},
there exist graphs of girth~$5$ and arbitrarily large degeneracy.
Hence $K_{i,j}$-free graphs are strictly more general than graphs
of bounded degeneracy. To the best of our knowledge, $K_{i,j}$-free
graphs form the largest class of graphs for which FPT algorithms and
polynomial kernels are known for the dominating set problem variants
discussed in this paper.

One interesting direction of future work is to try to demonstrate
kernels of size $f\left(d\right)\cdot k^{c}$ for the $k$\noun{-Dominating
Set} problem on $d$-degenerate graphs, where $c$ is independent
of $d$. Note that the result of Dom et al.\ mentioned above does
\emph{not} suggest that such kernels are unlikely. Another challenge
is to improve the running times of the kernelization algorithms: to
remove the exponential dependence on $d$ of the running time for
$d$-degenerate graphs, and to get a running time of the form $O(n^{c})$
for $K_{i,j}$-free graphs where $c$ is independent of $i$ and $j$.

\paragraph{Acknowledgments.}

We thank Aravind Natarajan for pointing out the connection between
$K_{i,j}$-free and $d$-degenerate graphs, Saket Saurabh and the
other authors of~\cite{DomLokshtanovSaurabh2008} for sharing with
us the lower-bound result mentioned in their paper, and Shai Gutner
for his comments on an earlier draft of this paper.

\bibliographystyle{acm}
\bibliography{ds_kernel}
 
\end{document}